%% file: dsm.tex
\documentclass[letterpaper, 12 pt]{article}

\usepackage{setspace}
\usepackage{cite}
\usepackage{graphicx}
\usepackage{color}
\usepackage{psfrag}
\usepackage{subfigure}
\usepackage{amsmath, amssymb, amsbsy, amsthm}
\usepackage{array}
\newtheorem{theorem}{Theorem}
\newtheorem{lemma}[theorem]{Lemma}

\newtheorem{assumption}{Assumption}
\newtheorem{corollary}{Corollary}
\newtheorem{proposition}{Proposition}

\usepackage{fullpage}

\newcommand{\newrj}[1]{#1}
\newcommand{\delrj}[1]{}

\renewcommand{\v}[1]{\ensuremath{\boldsymbol{\mathrm{#1}}}}
\renewcommand{\P}{\ensuremath{\mathcal{P}}}
\newcommand{\ignore}[1]{}
\newcommand{\E}{\ensuremath{\mathbb{E}}}

\newcommand{\PiZero}{\ensuremath{\Pi_{0}}}
\newcommand{\PiPrimary}{\ensuremath{\Pi_1}}

\newcommand{\PiSecondary}{\ensuremath{\Pi_2}}

\renewcommand{\baselinestretch}{1.2}

\begin{document}

\title{Competition in Wireless Systems via Bayesian Interference Games}
\author{Sachin Adlakha, Ramesh
Johari, and Andrea Goldsmith\thanks{This project was supported by the
Stanford Clean Slate Internet Program and by the Defense Advanced
Research Projects Agency under the Information Theory for Mobile Ad
Hoc Networks (ITMANET) Program.  The authors are with Stanford
University, Stanford, CA, 94305.}}
\date{August 27, 2007}

\maketitle


\begin{abstract}

We study competition between wireless devices with {\em incomplete
information} about their opponents. We model such interactions as {\em
Bayesian interference games}. \newrj{Each wireless device selects a power
profile over the entire available bandwidth to maximize its data rate
(measured via Shannon capacity), which requires mitigating the effect of
interference caused by other devices.} Such competitive models
represent situations in which several wireless devices share 
spectrum without any central authority or coordinated protocol.

In contrast to games where devices have complete information about
their opponents, we consider scenarios where the devices are unaware
of the interference they cause to other devices. Such games, which are
modeled as Bayesian games, can exhibit significantly different \delrj{Nash}
equilibria. We first consider a simple scenario where the devices
select their power profile simultaneously. In such
\textit{simultaneous move games}, we show that the unique \newrj{Bayes-}Nash
equilibrium is where both devices spread their power equally across
the entire bandwidth. We then extend this model to a two-tiered
spectrum sharing case where users act sequentially. Here one of the
devices, called the {\em primary user}, is the owner of the spectrum
and it selects its power profile first. The second device (called the
{\em secondary user}) then responds by choosing a power profile to
maximize its Shannon capacity. In such \textit{sequential move games},
we show that there exist \newrj{equilibria in} which the primary user
obtains a higher data rate by using only a part of the bandwidth.

In a {\em repeated} Bayesian interference game, we show the existence
of {\em reputation effects}: an informed primary user can ``bluff" to
prevent spectrum usage by a secondary user who suffers from lack of
information about the channel gains. The resulting equilibrium
\newrj{can be} highly inefficient, suggesting that competitive
spectrum sharing is highly suboptimal. This observation points to
the need for some regulatory protocol to attain a more efficient
spectrum sharing solution.

\end{abstract}

\input{Introduction}

\input{BGIG}

\input{Model}

\input{RepeatedGames}

\input{Conclusion}

\input{Appendix}

\renewcommand{\baselinestretch}{1}
\bibliographystyle{abbrv}
\bibliography{dsm}

\newpage
\begin{figure}[b]
\centering
\psfrag{User 1}[][][0.75]{User 1}
\psfrag{User 2}[][][0.75]{User 2}
\psfrag{h11}{$h_{11}$}
\psfrag{h21}{$h_{21}$}
\psfrag{h12}{$h_{12}$}
\psfrag{h22}{$h_{22}$}
\includegraphics[width=2in]{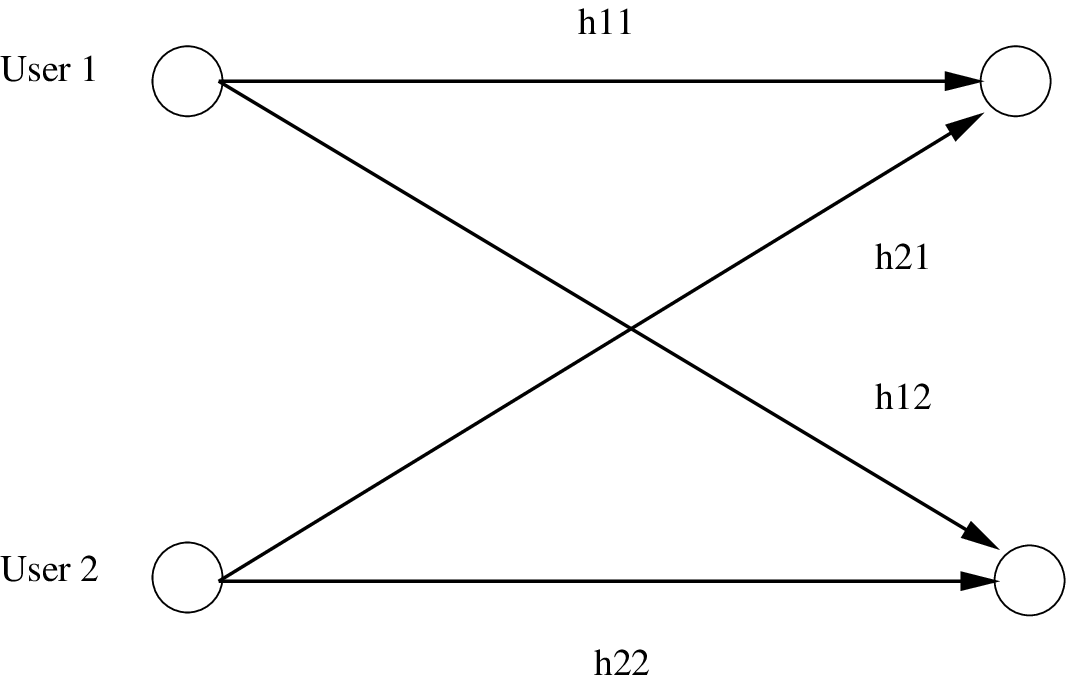}
\caption{Two User Interference Channel}
\label{Fig:InterferenceCh}
\end{figure}


\begin{figure}
\centering
\includegraphics[width=2.5in]{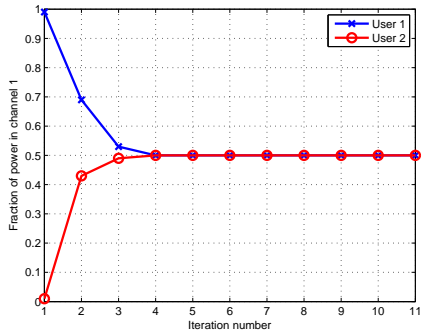}
\caption{{\em Convergence of best response dynamics for two user UC-GI
    game}.  We consider a model with $B = 2$ sub channels, $P = 1$,
    $N_0 = 0.01$, and we normalize $g_{11} = g_{22} = 1$.  We assume
    $g_{12}$ and $g_{21}$ are both drawn from a uniform distribution
    on $[0,1]$.  We initiate the best response dynamics at $P_{11} = 1
    - P_{12} = P$, and $P_{21} = 1 - P_{22} = 0$; observe that the
    powers $P_{ic}$ converge to $P/2$ for each $i$ and $c$.  The
    behavior is symmetric if we instead initiate with $P_{11} = 1 -
    P_{12} = 0$, and $P_{21} = 1 - P_{22} = P$.}
\label{Fig:BestResponse}
\end{figure}

\begin{figure}
\centering
\begin{psfrags}
\psfrag{X}{\footnotesize $X$}
\psfrag{N}{\footnotesize $N$}
\psfrag{SH}{\footnotesize $SH$}
\psfrag{SP}{\footnotesize $SP$}
\psfrag{S}{\footnotesize $2$}
\psfrag{P}{\footnotesize $1$}
\includegraphics[width=5in]{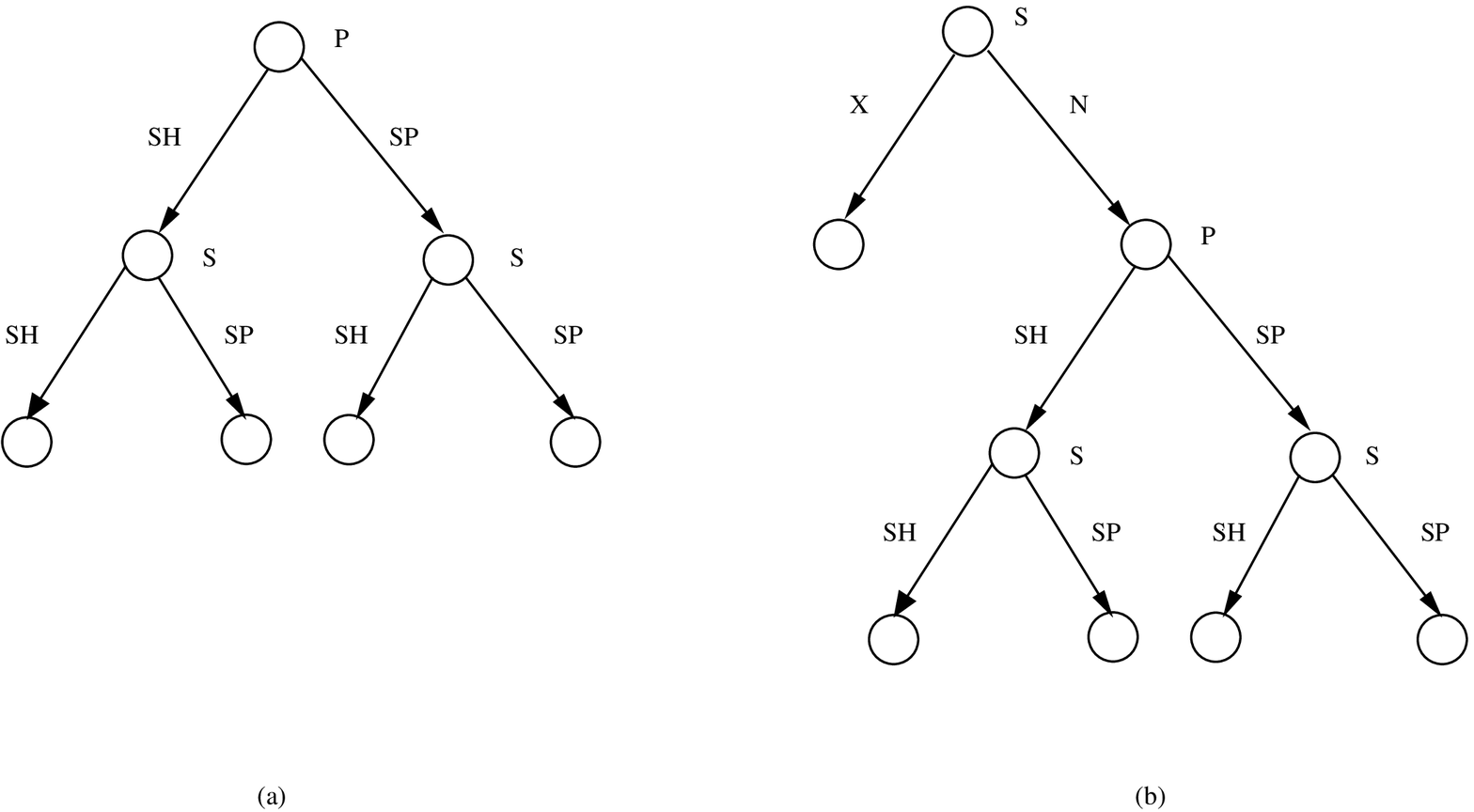}
\end{psfrags}
\caption{{\em Game trees for sequential games.}  Player 1 is the
  {\em primary} user; player 2 is the {\em secondary} user.  The tree in (a)
  describes the SBGI game.  The tree in (b) describes the SBGI-E game.}
\label{Fig:GameTreeCombined}
\end{figure}

\end{document}

%% file: Introduction.tex
\section{Introduction}
Our paper is motivated by a scenario where several wireless devices
share the same spectrum. Such scenarios are a common occurrence in
unlicensed bands such as the ISM and UNII bands. In such bands, diverse
technologies such as 802.11, Bluetooth, Wireless USB, and cordless
phones compete with each other for the same bandwidth. Usually, these
devices have different objectives, they follow different protocols,
and they do not cooperate with each other.  Indeed, although the FCC is
considering wider implementation of ``open'' spectrum sharing models,
one potential undesirable outcome of open spectrum could be a form of
the ``tragedy of the commons'': self-interested wireless devices
destructively interfere with each other, and thus eliminate potential
benefits of open spectrum.


Non-cooperative game theory offers a natural framework to model such
interactions between competing devices. In \cite{YGC_JSAC2002}, the
authors studied competition between devices in a Gaussian noise
environment as a {\em Gaussian interference} (GI) game. This work was
extended in \cite{EPT_JSAC2007} for the case of spectrum allocation
between wireless devices; the authors provided a 
non-cooperative game theoretic framework to study issues such as
spectral efficiency and fairness. In \cite{LL_DySpan2005}, the authors
derived channel gain regimes where cooperative schemes would perform
better than non-cooperative schemes for the GI game.

The game theoretic models used in these previous works typically
assume that the matrix of channel gains among all users is completely
known to the players.  This may not be realistic or practical in many
scenarios, as competing technologies typically do not employ a
coordinated information dissemination protocol.  Even if information
dissemination protocols were employed, incentive mechanisms would be
required in a situation with competitive devices to ensure that
channel states were truthfully exchanged.  By contrast, our paper
studies a range of non-cooperative games characterized by the feature
that there is {\em incomplete information} about some or all channel
gains between devices. Such scenarios are captured through static and
dynamic {\em Bayesian games} \cite{Fudenberg_tirole_1991}.

We consider a simplistic model where two transmitter-receiver (TX-RX)
pairs, or ``users'', share a single band divided into $K$
subchannels.\footnote{Throughout this paper, a transmitter-receiver
pair is identified with a particular user.}  We assume both users
face a total power constraint, and that the noise floor is identical
across subchannels.  We further assume that channel gains are drawn
from a fixed distribution that is common knowledge to the users.  We
make the simplifying assumption of {\em flat fading}, i.e., constant
gains across subchannels, to develop the model.
A user's strategic decision consists of an allocation of
power across the available subchannels \newrj{to maximize the available data
rate (measured via Shannon capacity).}

In Section \ref{sec:StaticGames} we consider a {\em simultaneous-move}
game between the devices under this model. We study two scenarios:
first, a game where all channel gains are unknown to both users; and
second, a game where a user knows the gain between its own TX-RX pair
as well as the interference power gain from \newrj{the other
transmitter at} its own receiver (also called incident channel gains),
but it does not know the channel gain between the TX-RX pair of the
other user or the interference it causes to \newrj{the other
receiver}.  In these two scenarios, we show that there exists a unique
symmetric Bayes-Nash equilibrium\footnote{Here ``symmetric'' means
that both users' strategies are identical functions of their channel
gain. Asymmetric equilibria, where users may have different functional
form of their strategies, are harder to justify, as they would require
prior coordination among the devices to ``agree'' on which equilibrium
is played.}, where both users equally spread their power over the band
(regardless of the channel gains observed). In this equilibrium, the
actions played after channel gains are realized are also a Nash
equilibrium of the complete information game.

While simultaneous-move games are a good model for competition between
devices with equal priority to shared resources, they are not
appropriate for a setting where one device is a natural incumbent,
such as primary/secondary device competition. In such two-tiered
models for spectrum sharing, some radio bands may be allocated to both
primary and secondary users. The primary users have priority over the
secondary users and we use game theory to analyze competition in such
scenarios.  In Section \ref{StackelbergGames}, we consider a two-stage
{\em sequential} Bayesian game where one device (the primary) moves
before the other (the secondary); we find that asymmetric equilibria
can be sustained where the devices sometimes operate in disjoint
subchannels \newrj{(called ``sharing'' the bandwidth)}, provided
interference between them is sufficiently large.  We also add an {\em
entry stage} to the game, where the secondary device decides whether
or not it wants to operate in the primary's band in the first place;
we also characterize Nash equilibria of this game in terms of the
distribution of the \newrj{incident} channel gains.

In Section \ref{sec:RepeatedGames}, we use the sequential Bayesian
game with entry to study {\em repeated} interaction between a primary
and secondary user.  We consider a model where a secondary user
repeatedly polls a primary user's band to determine if it is
worthwhile to enter. 
 Using techniques pioneered in the economics literature on {\em
  reputation effects} \cite{Mailath_samuelson_2006}, we show the
  existence of a {\em sequential equilibrium} where the primary user
  exploits the secondary user's lack of channel knowledge to its own
  advantage; in particular, we show that by threatening to
  aggressively spread power against the secondary, the primary can
  deter the secondary from entering at all.  (See footnote
  \ref{foot:seqeq} for the definition of sequential equilibrium.)
In a complete information game, the secondary user knows that the best
 response of the primary user to an entry by the secondary user is to
\newrj{share the bandwidth}. Thus, the primary user's threat of aggressively
 spreading power would not be \textit{credible} in such scenarios.
 Our result suggests that, in the absence of regulation, primary
 devices may inflate their power profile to ``scare'' secondary
 devices away, even if such behavior is suboptimal for the primary in
 the short term.

We conclude by noting that game theoretic models have also been used
in the {\em design} of power control and spectrum sharing schemes.  A
market-based power control mechanism for wireless data networks was
discussed in \cite{SMG_TCOMM2002}. In \cite{AA_TAC2003}, the authors
model a power control mechanism as a supermodular game and prove
several convergence properties.  Supermodularity was also employed to
describe a distributed power control mechanism in \cite{HBH_JSAC2006};
this latter paper also contains insights regarding supermodularity of
the GI game. Another approach to spectrum sharing is to consider real
time ``auctions"' of the channels as described in \cite{SMZ_JSAC2006,
HBH_MONET2006}.  By contrast, our paper studies a setting where no
coordination mechanism exists, and the devices are completely
competitive; hence we do not follow this approach.

%% file: BGIG.tex
\section{Static Gaussian Interference Games}
\label{sec:StaticGames}
In this section we consider a range of static game-theoretic models
for competition between two devices. In other words, in all the models
we consider, both devices {\em simultaneously} choose their actions,
and then payoffs are realized.  We start in Section \ref{sec:prelim}
by defining the model we consider, an interference model with two
users.  In Section \ref{sec:bgig} we define a Bayesian game where
both users do not know any of the channel gains.  However, this model
is not necessarily realistic; in many scenarios information is {\em
  asymmetric}: a device may know its own incident channel gains, but
not those incident on the other devices.  Thus in Section
\ref{sec:pbgig}, we introduce a ``partial'' Bayesian game (where 
the users know their own channel gain as well as the received interference gain), and study
its equilibria.  

\subsection{Preliminaries}
\label{sec:prelim}


We consider a two user Gaussian interference model (see Figure \ref{Fig:InterferenceCh}) with $K$
subchannels; in each subchannel $k = 1, \ldots, K$, the model is:
\begin{align}
y_{i}[n] = \sum_{j = 1}^{2}h^{k}_{ji}x_{j}[n] + w_{i}[n], \ i = 1, 2,
\label{Eq:InterferenceCh}
\end{align}
where $x_{i}[n]$ and $y_{i}[n]$ are user $i$'s input and output
symbols at time $n$, respectively. Here $h^{k}_{ij}$ is the channel gain
from the transmitter of user $i$ to the receiver of user $j$ in subchannel $k$. We assume that the system exhibits {\em
flat fading}, i.e., the channel gains $h^{k}_{ij} = h_{ij}$ for all $k = 1, \ldots, K$.  The noise processes $w_{1}[n]$ and $w_{2}[n]$ are assumed to be 
independent of each other, and are i.i.d over time with $w_{i}[n] \sim
\mathcal{N}(0, N_{0})$, where $N_{0}$ is the noise power spectral
density.

Each user has an average power constraint of $P$.  We assume that each
user treats interference as noise and that no interference
cancellation techniques are used.  Denote by $P_{ik}$ the transmission
power of user $i$ in channel $k$.  Let $\v{P}_i = (P_{i1}, \ldots, P_{iK})$,
and $\v{P} = (\v{P}_1, \v{P}_2)$.    We will frequently
use the notation ``$-i$'' to denote the player other than $i$ (i.e.,
player 1 if $i = 2$, and player 2 if $i = 1$).  The {\em utility}
$\Pi_{i}(\v{P})$ of user $i$ is the Shannon capacity data rate limit for the user.  Under the above assumptions, given a power
vector $\v{P}$, the Shannon capacity limit of a user $i$ over all $K$ subchannels is given
as:
\begin{align}
\Pi_i(\v{P}) = \sum_{k= 1}^{K} \left[\frac{1}{2}\log\left(1 + \frac{g_{ii}P_{ik}}{N_{0} + g_{-i,i}P_{ik}}\right)\right].
\label{Eq:UtilityFunctions}
\end{align}
Here $g_{ij}$ is the interference gain between the transmitter of user
$i$ and the receiver of user $j$, and is defined as $g_{ij} =
|h_{ij}|^{2}$; we let $\v{g} = (g_{11}, g_{12}, g_{21}, g_{22})$
denote the channel gain vector. Note
that for each $i$, the power allocation must satisfy the constraint
$\sum_k P_{ik} \leq P$.  In particular, both users share the same
power constraint.

In the complete information Gaussian interference (GI) game, each
user $i$ chooses a power allocation $\v{P}_i$ to maximize the utility
$\Pi_i(\v{P}_i, \v{P}_{-i})$ subject to the total power constraint,
given the power allocation $\v{P}_{-i}$ of the opponent.  Both users
carry out this maximization with full knowledge of the channel gains
$\v{g}$, the noise level $N_0$, and the
power limit $P$.  A {\em Nash equilibrium} (NE) of this
game is a power vector $\v{P}$ where both users have simultaneously
maximized payoffs.  This interference game has been analyzed
previously in the literature, and in particular existence and conditions
for uniqueness of the equilibrium have been developed in
\cite{YGC_JSAC2002, EPT_JSAC2007}.

In this paper, we take a different approach: we consider the same
game, but assume some or all of the channel gains are {\em unknown} to
the players.  In the next two sections, we introduce two variations on
this game.

\subsection{The Gaussian Interference Game with Unknown Channel Gains}
\label{sec:bgig}

We begin by considering the GI game, but where neither player has
knowledge of the channel gains $g_{ij}$; we refer to this as the {\em unknown
channel GI (UC-GI) game}. 
Our motivation is the fast-fading scenario where the channel gains change 
rapidly relative to the transmission strategy decision. This makes the 
channel gain feedback computationally expensive and generally inaccurate. 

We assume that the channel gains $\v{g}$ are drawn from a distribution
$F$, with continuous density $f$ on a compact subset $G \subset \{ \v{g} :
g_{ij} > 0\ \forall\ i,j\}$, and we assume that both players do not
observe the channel gains.  For simplicity, we assume that $F$ factors
so that $(g_{11}, g_{21})$ is independent of $(g_{22}, g_{12})$.

We assume that both players now maximize
{\em expected} utility, given the power allocation of their opponent;
i.e., given $\v{P}_{-i}$, player $i$ chooses $\v{P}_i$ to maximize
$\E[\Pi_i(\v{P}_i, \v{P}_{-i})]$, subject to the power constraint
(the expectation is taken over the distribution $F$).  A NE of the
UC-GI game is thus a power vector where both players have
simultaneously maximized their expected payoffs.

We focus our attention on the case of \textit{symmetric} NE, i.e.,
where both players use the same strategy.  It is possible that there
may exist several asymmetric equilibria, but for the users to operate at any 
one of those equilibria would require some form of prior
coordination. Since the users in this game do not coordinate, it is
reasonable to search for symmetric equilibrium. The next theorem shows
that if $K = 2$, the UC-GI game has a unique symmetric Nash
equilibrium. 

\begin{theorem}
For the UC-GI game with $K = 2$ subchannels, there exists a unique
symmetric pure strategy Nash equilibrium, regardless of the channel
distribution $F$, where the users spread their power equally over the
entire band; i.e, the unique NE is $P_{11}^* = P_{12}^* = P_{21}^* = P_{22}^*
= P/2$.
\label{Thm:UniquenessBGIG}
\end{theorem}

 \begin{proof}
 Note that if $\v{P}^*$ is a NE, then (substituting the power constraint)
 we conclude $P_{i1}^*$ is a solution of the following maximization problem:
 %
 \begin{align*}
 \max_{P_{i1}}\int_G \left[\frac{1}{2} \log\left(1 +
   \frac{g_{ii}P_{i1}}{N_0 + g_{-i,i}P_{-i,1}}\right) + \frac{1}{2}
   \log\left(1 + \frac{g_{ii}(P - P_{i1})}{N_0 + g_{-i,i}(P -
     P_{-i,1})}\right)\right]\; f(\v{g})\; d \v{g}. 
 \end{align*}
 Since $\log(1 + x)$ is strictly concave in $x$, the first order
 conditions are necessary and sufficient to identify a NE.  
 Differentiating and simplifying yields:
 \[ \int_G \frac{g_{ii}}{2}\left(\frac{g_{ii}(P-2P_{i1}) +
   g_{-i,i}(P-2P_{-i,1})}{(N_0 + g_{ii}P_{i1} + g_{-i,i}P_{-i,1})(N_0 +
 g_{ii}(P - P_{i1} +  g_{-i,i}(P-P_{-i,1})))}\right)\; f(\v{g}) \;
   d\v{g} = 0. \]

 Note that the denominator in the integral above is always
 positive; and further, $g_{ii} > 0$ on $G$.  Thus in a NE, if $P_{i1}
 > P/2$, then we must have $P_{-i,1} < P/2$ (and vice versa).  Thus the
 only symmetric NE occur where $P_{i1} = P_{-i,1} = P/2$, as required.
 \end{proof}


While our result is framed with only two subchannels, the same
argument can be easily extended to the case of multiple subchannels
via induction. 



\begin{corollary}
Consider the UC-GI game with $K > 1$ subchannels.  There exists a
unique symmetric NE, where the two users spread their power equally
over all $K$ subchannels.
\label{Corr:UniquenessBCh}
\end{corollary}

 \begin{proof}
 The proof follows from an inductive argument; clearly the result holds
 if $K = 2$.  Let $\v{P}^K$ be a
 symmetric NE with $K$ subchannels.  Let $S \subset \{ 1, \ldots, K \}$
 be a subset of the subchannels.  Since the NE is symmetric, let
 $Q^S = \sum_{k \in S} P_{ik}^K$; this is the total power the
 players use in the  subchannels of $S$.  It is clear that if we
 {\em restrict} the power vector $\v{P}^K$ to only the 
 subchannels in $S$, then the resulting power vector must be a
 symmetric NE for the UC-GI game over only these 
 subchannels, with total power constraint $Q^{K-1}$.  Since this holds
 for every subset $S \subset \{ 1, \ldots, K \}$ of size $|S| \leq
 K-1$, we can apply the inductive hypothesis to conclude every user
 allocates equal power to each subchannel in the equilibrium $\v{P}^K$, as required.
 \end{proof}


While we have only shown uniqueness among {\em symmetric} NE in the
preceding results, we conjecture that in fact the {\em only} pure NE
of the UC-GI game is one where all players use equal transmit power in
every subchannel.  Our conjecture is motivated by numerical results
using {\em best response dynamics} for the UC-GI game; these are
dynamics where at each time step, each player plays a best response to
the action of his opponent at the previous time step.  As we see in
Figure \ref{Fig:BestResponse}, even if the users initially transmit at
different powers in each subchannel, the best response dynamics converge
to the symmetric NE.  In fact, for this numerical example the best
response dynamics verify uniqueness of the pure NE.\footnote{For the
UC-GI game, one can infer that for the numerical example with $K = 2$,
the unique pure NE is the symmetric NE where all users spread their
power across the subchannels.  To justify this claim, note that the
UC-GI game is a {\em supermodular game}, if the strategy spaces are
appropriately defined.  (A complete overview of supermodular games is
beyond the scope of this paper; for background on supermodular games,
see \cite{Topkis_1998, HBH_JSAC2006}.)  In particular, let $s_1 =
P_{11}$, and let $s_2 = -P_{21}$, with strategy spaces $S_1 = [0,P]$,
$S_2 = [-P,0]$.  Define $V_i(s_1, s_2) = \Pi_i(s_1, P - s_1, -s_2, P +
s_2)$.  Then it can be easily shown that $V_i$ has increasing
differences in $s_i$ and $s_{-i}$. This suffices to ensure that there
exists a ``largest'' NE $\overline{s}$, and a ``smallest'' NE
$\underline{s}$, that are, respectively, the least upper bound and
greatest lower bound to the set of NE in the product lattice $S_1
\times S_2$ \cite{MR_Econ1990}.  Further, best response dynamics
initiated at the smallest strategy vector $(s_1, s_2) = (0,-P)$
converge to $\underline{s}$; and best response dynamics initiated at
the largest strategy vector $(s_1, s_2) = (P,0)$ converge to
$\overline{s}$ \cite{MR_Econ1990}. Thus if these two best response
dynamics converge to the same strategy vector, there must be a unique
pure NE.}

\subsection{Bayesian Gaussian Interference Game}
\label{sec:pbgig}
In the UC-GI game defined above, we assume that each user is unaware
of all the channel gains. However in a slowly changing environment, it is common for the receiver to feed back
channel gain information to the transmitter. Thus, in this section 
we assume that each user $i$ is aware of the self channel gain $g_{ii}$,
and the incident channel gain $g_{-i,i}$; for notational simplicity,
let $\v{g}_i = (g_{ii}, g_{-i,i})$.  
However, because of the difficulties involved in dissemination of
channel state information from other devices, we continue to assume
that each user is unaware of the channel gains of the other users.  In
particular, this means user $i$ does not know the value $\v{g}_{-i}$.
In this game, each player chooses 
$\v{P}_i$ to maximize $\E[\Pi_i(\v{P}_i,
  \v{P}_{-i})|\v{g}_i]$, subject to the power constraint; note
that now the expectation is conditioned on $\v{g}_i$.
The power allocation $\v{P}_{-i}$ is random, since it depends on the
channel gains of player $-i$---which are unknown to player $i$.
Thus this is a {\em Bayesian game}, in which a strategy of player $i$ is
a family of functions $\v{s}_i(\v{g}_i) = (s_{i1}(\v{g}_i), \ldots,
s_{iK}(\v{g}_i))$, 
where $s_{ik}(\v{g}_i)$ gives the power allocation of player $i$ in
subchannel $k$ when gains $\v{g}_i$ are realized.  We refer to this game as the
{\em Bayesian Gaussian interference (BGI) game}.  A {\em Bayes-Nash
  equilibrium} BNE is a strategy vector $(\v{s}_1(\cdot),
\v{s}_2(\cdot))$ such that for each $i$ and each  
$\v{g}_i$, player $i$ has maximized his expected payoff given the
strategy of the opponent:
\[ \v{s}_i(\v{g}_i) \in \arg \max_{\v{P}_i} \E[\Pi_i(\v{P}_i,
  \v{s}_{-i}(\v{g}_{-i}))|\v{g}_i]. \]



For the BGI game, we again want to investigate symmetric Bayes-Nash
equilibria.  However, in principle the functional strategic form of a
player can be quite complex.  Thus, for analytical tractability, we
focus our attention on a restricted class of possible actions: we
allow users to either put their entire power in a single subchannel,
or split their power evenly across all subchannels.  This is a
practical subclass of actions which allows us to explore
whether asymmetric equilibria can exist.

Formally, the action space of both players is now restricted to $S =
\{ P \v{e}_1, \ldots, P \v{e}_K, P\v{1}/K \}$; here $\v{e}_i$ is the
standard basis vector with all zero entries except a ``1'' in the
$i$'th position, and $\v{1}$ is a vector where every entry is ``1''.
Thus $P \v{e}_k$ is the action that places all power in subchannel
$k$, while $P \v{1}/K$ spreads power equally across all subchannels.
A strategy for player $i$ is a
map that chooses, for each realization of $(g_{ii}, g_{-i,i})$, an
action in $S$. 

Our main result is the following theorem.

\begin{theorem}
Assume that $(g_{11}, g_{21})$ and $(g_{22}, g_{12})$ are i.i.d. Then
the unique pure strategy symmetric BNE of the BGI game is where both
users choose action $P \v{1}/K$, i.e., they spread their power
equally across bands.
\label{Thm:P-BGIG}
\end{theorem}

\begin{proof}
Fix a symmetric BNE $(\v{s}_1, \v{s}_2)$ where $s_{1k}(\cdot) =
s_{2k}(\cdot) = s^k(\cdot)$ is the common strategy used by both players;
i.e., given channel gains $\v{g}_i$, player $i$ puts power
$s^k(\v{g}_i)$ in subchannel $k$.  Define $\alpha_k = \P(s^k(\v{g}_i)
= P \v{e}_k)$ for each subchannel $k$; and $\gamma = \P(s^k(\v{g}_i) =
P \v{1}/K)$.  These are the probabilities that a player
transmits with full power in subchannel $k$, or with equal power in
all subchannels, respectively. 

Let $\overline{\Pi}_i(\v{P}_i ; \v{g}_i)$ be the expected payoff of user $1$ if it uses
action $\v{P}_i \in S$, given that the other player is using the
equilibrium strategy profile $(s^1, \ldots, s^K)$ and the channel
gains are $\v{g}_i$.  We start with the following lemma. 

\begin{lemma}
For two subchannels $k, k'$, if $\alpha_k < \alpha_{k'}$, then 
$\overline{\Pi}_i(P \v{e}_k ; \v{g}_i) > \overline{\Pi}_i(P \v{e}_{k'}
; \v{g}_i)$ for all values of $\v{g}_i$; i.e., player $i$
strictly prefers to put full power into subchannel $k$ over putting full
power into subchannel $k'$.
\label{Lemma:PayoffComparison}  
\end{lemma}

{\em Proof of Lemma}.  
Using \eqref{Eq:UtilityFunctions} we can write $\overline{\Pi}_{i}(P \v{e}_k
; \v{g}_i)$ as: 
\begin{align*} 
\overline{\Pi}_{i}(P \v{e}_k ; \v{g}_i) &= \frac{\alpha_k}{2}\left[\log\left(1 +
  \frac{g_{ii}P}{N_{0} + g_{-i,i}P}\right)\right] +
\frac{\gamma}{2}\left[\log\left(1 + \frac{g_{ii}P}{N_{0} +
    g_{-i,i}P/B}\right)\right] \\ 
&+ \frac{1 - \alpha_k - \gamma}{2}\left[\log\left(1 +
  \frac{g_{ii}P}{N_{0}}\right)\right].
\end{align*}
Define $\Delta$ as:
\begin{align*}
\Delta \triangleq \frac{1}{2}\log\left(1 +
\frac{g_{ii}P}{N_{0}}\right) - \frac{1}{2}\log\left(1 +
\frac{g_{ii}P}{N_{0} + g_{-i,i}P}\right) > 0,
\end{align*}
since we have assumed $g_{-i,i} > 0$.  Rearranging and simplifying, we have that 
$\overline{\Pi}_{i}(P \v{e}_k ; \v{g}_i) - \overline{\Pi}_{i}(P\v{e}_{k'} ; \v{g}_i) = \Delta\left(\alpha_{k'} - \alpha_k\right).$ 
Since $\Delta > 0$ and $\alpha_{k'} > \alpha_k$, the lemma is proved.\hfill$\Box$

The previous lemma ensures that in a symmetric equilibrium we cannot
have $\alpha_k < \alpha_{k'}$ for any two subchannels $k, k'$: in this
case, $\alpha_{k'} > 0$, so 
the equilibrium strategy 
puts positive weight on action $P \v{e}_{k'}$; but player $i$'s best response
to this strategy puts zero weight on $P \v{e}_{k'}$ (from the lemma).

Thus in a symmetric equilibrium we must have $\alpha_k = \alpha_{k'}$
for all subchannels $k, k'$; i.e., $\alpha_k = (1 -\gamma)/K$ for all
$k$.  Define $\alpha = (1 - \gamma)/K$.  In this case we have for each
subchannel $k$:
\begin{align*}
\overline{\Pi}_{i}(P \v{e}_k ; \v{g}_i) =
&\alpha\left[\frac{K-1}{2}\log\left(1 + 
\frac{g_{ii}P}{N_{0}}\right) + \frac{1}{2}\log\left(1 +
\frac{g_{ii}P}{N_{0} + g_{-i,i}P}\right)\right] \nonumber \\
 &+ \gamma\left[\frac{1}{2}\log\left(1 + \frac{g_{ii}P}{N_{0} +
g_{-i,i}P/K}\right)\right],
\end{align*}
and 
\begin{align*}
\overline{\Pi}_{i}(P \v{1}/K) =& (1-\gamma)
\left[\frac{K-1}{2}\log\left(1 + \frac{g_{ii}P/K}{N_{0}}\right) +
  \frac{1}{2}\log\left(1 + \frac{g_{ii}P/K}{N_{0} +
    g_{-i,i}P}\right)\right]\\ 
& + \gamma\left[\frac{K}{2}\log\left(1 + \frac{g_{ii}P/K}{N_{0} + g_{-i,i}P/K}\right)\right].
\end{align*}
Since $\log(1 + x)$ is a strictly concave function of $x$, we get that $K\log(1
+ \frac{x}{K}) > \log(1 + x)$ for $x > 0$.  Since $\alpha = (1 -
\gamma)/2$, this implies that
$\overline{\Pi}_{i}(P \v{1}/K) > \overline{\Pi}_{i}(P \v{e}_k)$ for
all subchannels $k$.  Thus in a symmetric equilibrium, we must have
$\alpha = 0$; i.e., the unique symmetric equilibrium occurs where
$\gamma =1$, so both users equally spread their power across all
subchannels.
\end{proof}


Thus far we have considered games where players act simultaneously. However in several
practical cases, one of the players may already be using the spectrum when another 
user wants to enter the same band. We model such scenarios as sequential games in the next section.

%% file: Model.tex
\section{Sequential Interference Games with Incomplete Information}
\label{StackelbergGames}

In this section, we study sequential games between wireless devices
with incomplete information.  In such games, player 1, who
we refer to as the \textit{primary user}, determines its transmission
strategy before player 2.  Player 2, also referred to
as the \textit{secondary user}, observes the action of the primary
user and chooses a transmission strategy that is a best response.  We
study this model in Section \ref{sec:2StageGame}.  While we focused on
symmetric equilibria of static games in the preceding section, we
focus here on the fact that sequential games naturally allow the users
to sustain asymmetric equilibrium.  We characterize how these
equilibria depend on the realized channel gains of the users.

Such games are a natural approach to study dynamic spectrum sharing between
cognitive radios.  The primary user is the incumbent user of the band,
while the secondary user represents a potential new device that also
wishes to use spectrum in the band.  In particular, in this model we
must also study {\em whether} the secondary device would find it
profitable to compete with the primary in the first place. In Section
\ref{sec:SeqGamesEntry}, we extend this game to incorporate an entry
decision by the secondary user, and again study dependence of the
equilibria on realized channel gains.

\subsection{A Two-Stage Sequential Game}
\label{sec:2StageGame}

In this section, we consider a two-stage \textit{sequential Bayesian
Gaussian interference (SBGI)} game; we restrict attention to two
subchannels for simplicity.  Player 1 (the primary) moves first, and
Player 2 (the secondary) {\em perfectly} observes the action of the
primary user before choosing its own action.  We assume that each of
the self gains $g_{ii}$ (in the interference channel given in
\eqref{Eq:InterferenceCh}) are normalized to $1$; this is merely done
to isolate and understand the effects of interference on the
interaction between devices, and does not significantly constrain the
results.

As before, the channel gains are randomly selected.  For the
remainder of this section, we make the following assumption. 
\begin{assumption}
Player $1$ (the primary user) knows both $g_{12}$ and $g_{21}$, while the
secondary user only knows $g_{12}$ (but not $g_{21}$).
\label{As:informationSets}
\end{assumption}
Thus the primary user knows the interference it causes to the
secondary user; the secondary user however, is only aware of its own
incident channel gain.  This approach allows us to focus on the
secondary user's uncertainty; it is also possible to analyze the same
game when the primary user does not know $g_{12}$ (see appendix).

As before, we restrict the action space of each user to either use
only one of the subchannels, or to spread power equally in both
subchannels.  We assume that if the primary concentrates power in a
single subchannel, it chooses subchannel $1$; this is done without
loss of generality, since fading is flat, and the primary moves first.
If the primary user chooses to place its entire power in subchannel
$1$, then from Lemma \ref{Lemma:PayoffComparison}, we know that the
secondary user's best response puts zero weight on this action: the
secondary user will either spread over both subchannels, or put its
entire power in the free subchannel.  Thus concentrating power in a
single subchannel is tantamount to {\em sharing} a single subchannel
with the secondary.  Thus we say the primary ``shares'' if it
concentrates all its power in a single subchannel, and denote this
action by $SH$.  We use $SP$ to denote the action where the primary
spreads its power across both subchannels.  We also use the same
notation to denote the actions of the secondary: concentrating power
in a single subchannel  (in this case, subchannel $2$) is denoted by 
$SH$, and spreading is denoted by $SP$. The game tree for the SBGI game 
is shown in Figure \ref{Fig:GameTreeCombined}.



We solve for the equilibrium path in the sequential game using
backward induction.  Once channel gains are realized, suppose that the
primary player chooses the action $SP$.  Conditioned on this action by
the primary player, the secondary player chooses the action $SP$ if
\begin{align}
\log\left(1 + \frac{P/2}{N_{0} + g_{12}P/2}\right) > \frac{1}{2}\log\left(1 + \frac{P}{N_{0} + g_{12}P/2}\right).
\end{align}
Since $\log(1 + x)$ is a strictly concave function of $x$, the above
inequality holds for all values of $g_{12}$. Thus the best
response of the secondary user is to choose $SP$ whenever the primary
user chooses $SP$, regardless of the value of $g_{12}$.

The situation is more interesting if the primary user decides to share
the bandwidth, i.e., chooses action $SH$.  Given this action
of the primary user, the secondary user would prefer to spread its
power if and only if:
\begin{align}
\frac{1}{2}\log\left(1 + \frac{P}{2N_{0}}\right) &+ \frac{1}{2}\log\left(1 + \frac{P/2}{N_{0} + g_{12}P}\right) > \frac{1}{2}\log\left(1 + \frac{P}{N_{0}}\right) \ \Longleftrightarrow \ g_{12} < \frac{1}{2}.
\label{Eq:g12Threshold}
\end{align}
Thus if $g_{12} < 1/2$, the secondary will choose $SP$ in response to
$SH$.  On the other hand, if $g_{12} > 1/2$, the secondary user will
choose $SH$ in response to $SH$.  (We ignore $g_{12} = 1/2$ since
channel gains have continuous densities.)  We summarize our
observations in the next lemma.

\begin{lemma}
Suppose Assumption \ref{As:informationSets} holds.  In the SBGI game, if the
primary user chooses $SP$ in the first stage, the best response of the
secondary is $SP$; if the primary chooses $SH$ in the first stage, the
best response of the secondary is $SH$ if $g_{12} > 1/2$, and $SP$ if
$g_{12} < 1/2$.
\label{Lemma:SBGI}
\end{lemma}

Thus in equilibrium, regardless of the channel gains, either both
users share or both users spread.  Suppose the secondary plays $SP$
regardless of the primary's action; in this case $SP$ is also the best
action for the primary.  Thus if $g_{12} < 1/2$, the primary user will
never choose $SH$ in the first stage.  On the other hand, when $g_{12}
> 1/2$, it can choose its optimal action by comparing the payoff when
{\em both} players choose $SP$ to the payoff when {\em both} players
choose $SH$.  In this case, the primary user would prefer to spread
its power if and only if:
\begin{align}
\!\log\left(1 + \frac{P/2}{N_{0} + g_{21}P/2}\right) &> \frac{1}{2}\log\left(1 + \frac{P}{N_{0}}\right), \ 
\Longleftrightarrow\ g_{21} < \frac{1}{\sqrt{1 + P/N_{0}} - 1} &- \frac{2N_{0}}{P} \triangleq g^{*}.
\label{Eq:g21Threshold}
\end{align}

Observe that this threshold approaches zero as $P/N_0 \to \infty$, and
$1/2$ as $P/N_0 \to 0$. (In fact, the threshold is a decreasing function of $P/N_0$.)
 Thus we have the following proposition.


\begin{proposition}
Suppose Assumption \ref{As:informationSets} holds. In the sequential equilibria of the game,\footnote{\label{foot:seqeq}Sequential equilibrium is a standard solution
concept for dynamic games of incomplete information
\cite{KW_ECONOMETRICA1982}.  Sequential equilibrium consists of two
elements for each user: a history- and type-dependent strategy, as
well as a conditional distribution (or {\em belief}) over the unknown types of other
players given history.  Two conditions must be satisfied; first, for
each player, the strategy maximizes expected payoff over the remainder
of the game (i.e., the strategy is \textit{sequentially rational}).
Second, the beliefs are {\em consistent}: players compare observed history
to the equilibrium strategies, and use Bayesian updating to specify their conditional
distribution over other players' types.  The precise definition is
somewhat more involved, and beyond the scope of this paper.

In our scenario, the primary knows $g_{12}$, and thus has no
uncertainty.  The secondary user's belief of the value of $g_{21}$ is updated at the second
stage on the basis of the primary's action at the first stage;
however, since the secondary user's action does not depend on the
value of $g_{21}$, we do not specify beliefs in the statement of the proposition.}
 if $g_{12} > 1/2$ and $g_{21} > g^{*}$, the
 primary user chooses $SH$; if $g_{12} < 1/2$ or $g_{21} < g^*$, the primary user chooses $SP$.
 The secondary user plays the same action as that chosen by the
 primary, regardless of the realized channel gain.
 The value of $g^{*}$ is given in \eqref{Eq:g21Threshold}.
\label{Thm:NE-SGI}
\end{proposition}


%

Since in equilibrium, either both users share or both users spread,
for later reference we make the following definitions:
\begin{align}
\Pi_i^{\text{share}} = \frac{1}{2} \log \left( 1 + \frac{P}{N_0}
\right); 
\quad 
\Pi_i^{\text{spread}} = \log \left( 1 + \frac{P/2}{N_0 + g_{-i,i} P/2}
\right). \label{Eq:JointRates2}
\end{align}
These are the payoffs to user $i$ if {\em both} users play $SH$, or
both play $SP$, respectively.  Note that $\Pi_i^{\text{spread}}$
depends on $g_{-i,i}$, and is thus stochastic.

\subsection{A Sequential Game with Entry}
\label{sec:SeqGamesEntry}

In this section, we modify the game of the last section
to incorporate an entry decision by the secondary user; we refer to
this as the \textit{sequential Bayesian Gaussian interference game with entry
(SBGI-E)}. Specifically, at the beginning of the game, the secondary
user decides to either enter the system (this action is denoted by
$N$) or stay out of the system (denoted by $X$). If the secondary user
decides to exit, the primary user has no action to take. If however,
the secondary user enters the system, the two users play the same game
as described in the previous section. The game tree of the SBGI-E game
is shown in Figure \ref{Fig:GameTreeCombined}.

If the secondary user exits the game, its payoff (defined as the
maximum data rate received) is $0$. However, if the secondary user
enters the game, the payoff to the user is always positive, regardless
of the action taken by the primary user and the channel gains. Thus,
in the absence of any further assumption, the model trivially reduces to
the one studied in the previous section. To make the model richer, we
introduce a {\em cost of power} to the overall payoff of the
secondary user.\footnote{This cost of power can also be introduced for
the primary user without changing any results mentioned in this
paper. We avoid this for the sake of simplicity.}  This cost of power
can represent, for example, battery constraints of the wireless
device.  We assume that if the secondary enters, a cost $kP$ is
incurred (where $k$ is a proportionality constant). With this cost of
power, if the secondary user enters it obtains payoff 
$\hat{\Pi}_{2}(\v{P}_1, \v{P}_2) = \Pi_2(\v{P}_1, \v{P}_2)- kP$. 
Note that if the secondary user exits the game, it gets no rate with
no cost of power; in this case its payoff is zero.  Furthermore, from
Proposition \ref{Thm:NE-SGI}, we know that if the 
secondary user enters in equilibrium, it will obtain either
$\Pi_2^{\text{share}} - kP$ or $\Pi_2^{\text{spread}} -kP$.  Thus, to decide
between entry and exit, the secondary user compares these quantities
to zero.  In the case where $g_{12} < 1/2$, we easily obtain the
following proposition.

\begin{proposition}
Suppose Assumption \ref{As:informationSets} holds, and that $g_{12}< 1/2$.  In
the sequential equilibrium of the
SBGI-E game, the secondary player always enters if
$\Pi_2^{\text{spread}} > kP$ and it always exits if
$\Pi_2^{\text{spread}} < kP$. 
\end{proposition}
\begin{proof}
The proof follows trivially from Theorem \ref{Thm:NE-SGI} since if
$g_{12} < 1/2$, both the primary and the secondary user spread their
powers after entry.
\end{proof}

Let us now consider the case when $g_{12} > 1/2$.  Since the game is
symmetric, we conclude that $\Pi_2^{\text{share}} >
\Pi_2^{\text{spread}}$ if and only if $g_{12} > g^*$.  A
straightforward calculation using the expression in
\eqref{Eq:g21Threshold} establishes that $g^* < 1/2$.  Thus if $g_{12}
> 1/2$, then $g_{12} > g^*$; in particular, for $g_{12} > 1/2$, there
always holds $\Pi_2^{\text{share}} >
\Pi_2^{\text{spread}}$. For the sake of simplicity, we make the following assumption for the
remainder of the paper.




\begin{assumption}
The payoff to the secondary user $\Pi_2^{\text{share}}$ is greater
than the cost of power $kP$; i.e., $P/N_0 > 2^{2kP} - 1$.
\label{As:SignOfb}
\end{assumption}
We now compare $\PiSecondary^{\text{spread}}$ to the cost of power
$kP$. Under Assumption \ref{As:SignOfb}, the secondary user would
always enter if $\PiSecondary^{\text{spread}} > kP$.  This happens if
and only if:
\begin{align}
\log\left(1 + \frac{P/2}{N_{0} + g_{12}P/2}\right) > kP \quad 
\Longleftrightarrow \ g_{12} < \frac{1}{2^{kP} - 1} - \frac{2N_{0}}{P} \triangleq \tilde{g}_{12}.
\label{Eq:g12Tilde}
\end{align}
Thus if $1/2 < g_{12} < \tilde{g}_{12}$, the secondary user always enters.  
Note that $\tilde{g}_{12} \to -\infty$ if $N_0 \to \infty$, for fixed
$P$. For fixed $N_0$, we have $\tilde{g}_{12} \to 0$ as $P \to
\infty$; and a straightforward calculation\footnote{To see this, let $c = 2 k N_0 \ln 2$ and $x = P/(2N_0)$, and note
that:
\[ \tilde{g}_{12} = \frac{1}{e^{cx} - 1} - \frac{1}{x}. \]
Note that $e^{cx} = 1 + cx + c^2x^2/2 + e(x)$, where $e(x) = o(x^2)$.
The result follows by substituting this Taylor expansion in the above
expression for $\tilde{g}_{12}$, and considering the cases $c < 1$, $c
> 1$, and $c = 1$, respectively.}
shows that:
\[
\lim_{P \to 0} \tilde{g}_{12} =\left\{ \begin{array}{ll}
\infty,& \ \text{if}\ 2kN_0 \ln 2 < 1;\\
-\infty,&\ \text{if}\ 2kN_0 \ln 2 > 1;\\
-1/2,&\ \ \text{if}\ 2kN_0 \ln 2 = 1.
\end{array}
\right. \]
In particular, $\tilde{g}_{12}$ can take any real value depending on
the parameters of the game.  

If $g_{12} > \tilde{g}_{12}$, then $\Pi_2^{\text{spread}} < kP <
\Pi_2^{\text{share}}$, so the secondary would only enter if
the primary user shares the channel. Let $\rho = \P(g_{21} <
g^{*})$ be the probability that the primary user spreads the
power. Then the expected payoff of the secondary user on entry is 
$\overline{\Pi}_{2} = \rho \PiSecondary^{\text{spread}} + (1-\rho)\PiSecondary^{\text{share}} - kP$.
The secondary user would enter if its expected payoff is
positive. This happens if  and only if:
\begin{align}
\rho < \frac{\Pi_2^{\text{share}} - kP}{\PiSecondary^{\text{share}} - \PiSecondary^{\text{spread}}} \triangleq d.
\label{Eq:ProbCutoff}
\end{align}

The equilibria of the SBGI-E for $g_{12} > 1/2$ are summarized in the
following proposition.   
\begin{proposition}
Suppose Assumptions \ref{As:informationSets} and \ref{As:SignOfb}
hold, and that $g_{12} > 1/2$.  Define $d$ as in
\eqref{Eq:ProbCutoff}, and $\rho = \P(g_{21} < g^*)$. Then in the
sequential equilibrium of the game, if $g_{12} < \tilde{g}_{12}$, the secondary user always enters the game; 
if however $g_{12} > \tilde{g}_{12}$ the secondary user enters the game if $\rho < d$, and it exits the game if $\rho > d$. Upon entry the primary and the secondary user follow the sequential equilibrium of the SBGI game as given in Proposition \ref{Thm:NE-SGI}.
\label{Thm:NE-SGI-E}
\end{proposition}


%% file: RepeatedGames.tex
\section{Repeated Games with Entry: The Reputation Effect}
\label{sec:RepeatedGames}

In this section, we study {\em repeated} interactions between wireless
devices with incomplete information about their opponents.  We
consider a finite horizon repeated game, where in each period the
primary and secondary users play the SBGI-E game studied in the
previous section.  Consider, for example, a single secondary device
considering ``entering'' one of several distinct bands, each owned by
a different primary.  The secondary is likely to {\em poll} the
respective bands of the primaries, probing to see if entering the band
is likely to yield a high data rate.  Each time the secondary probes a
single primary user's band, it effectively decides whether to enter or
exit; we model each such stage as the SBGI-E game of the preceding
section.

For analytical simplicity, we assume that the secondary
user is {\em myopic}; i.e., it tries to maximize its single period
payoff.  The primary user acts to maximize its total
(undiscounted) payoff over the entire horizon. Even though the
secondary user is myopic, it has \textit{perfect} recall of the
actions taken by both the primary and the secondary user in previous
periods of the game.  We find that this model can be studied using
seminal results from the economic literature on {\em reputation
  effects}; in particular, our main insight is that the primary may
choose to spread power against an entering secondary, even if it is
not profitable in a single period to do so---the goal being to
``scare'' the secondary into never entering again.
We first describe our repeated game model in section
\ref{sec:RepeatedModel}. Then, in section \ref{sec:ReputationResults},
we analyze sequential equilibria for such games.

\subsection{A Repeated SBGI-E Game}
\label{sec:RepeatedModel}

We first assume that at the beginning of the repeated game, nature
chooses (independent) cross channel gains $g_{12}$ and $g_{21}$ from a
known common prior distribution $F$, and these channel gains stay
constant for the entire duration of the game.  As in the preceding
section, we assume $g_{11} = g_{22} = 1$, to isolate the effect of
interference.  We continue to make the assumption (for technical
simplicity) that the secondary user is not aware of $g_{21}$, while
the primary user knows the value of $g_{12}$; insights for the case
where the primary does not know $g_{12}$ are offered in the appendix
As before, the channel gains determine the
data rate obtained by each user.

We assume that the primary and the secondary user play the same SBGI-E game
in each period; i.e., each period the secondary decides whether to
enter or exit. Because the secondary player is assumed to be myopic,
once the secondary player enters the game its best response to the
primary user's action is uniquely defined by Proposition
\ref{Thm:NE-SGI}. In particular, post-entry, the best response of the
secondary is {\em identical} to the action taken by the
primary---regardless of the channel gains. Thus, we can reduce the
three-stage SBGI-E game into a two-stage SBGI-E game. In the first
stage of this reduced game, the secondary user chooses either $N$
(enter) or $X$ (exit). If the secondary user enters, then in the
second stage, the primary user chooses either $SH$ (share) or $SP$
(spread).  The payoffs of the players are then realized, using the
fact that the post-entry action of the secondary user is same as the
action of the primary user.

We let $(a_{1,t}, a_{2,t})$ denote the actions chosen by the two players at each time period $t$.
If the secondary exits, then the primary is always
strictly better off spreading power across subchannels instead of
concentrating in a single subchannel.  Thus $a_{2,t} = X$ is never
followed by $a_{1,t} = SH$ in equilibrium, so without loss of
generality we assume $(a_{1,t}, a_{2,t}) \in \{ (SP, X), (SH,
N), (SP, N) \}$.  


We assume both users have \textit{perfect} recall of the actions taken by the users in
previous periods. Let $h_{i,t}$ denote the {\em history} recalled by
player $i$ in period $t$. Then, $h_{1,t} = (\v{a}_{1}, \cdots
\v{a}_{t-1}, a_{2,t})$ and $h_{2,t} = (\v{a}_{1}, \cdots \v{a}_{t-1})$
(since the primary observes the entry decision of the secondary before
moving). The strategy $s_{i}(h_{i,t})$ of a user~$i$ is a
probability distribution over available actions ($\{X, N\}$ for the
secondary user, and $\{SH, SP\}$ for the primary user).  
Note in particular that in this section we allow {\em mixed
  strategies} for both players.

By a slight abuse of notation, let $\Pi_{i}(\v{a}_{t})$ denote the
payoff of player $i$ in period $t$. For the primary user, the payoff
in each period is the maximum data rate it gets in that period, and
its objective is to maximize the total payoff $\sum_{t=
1}^{T}\Pi_{1}(\v{a}_t)$. Here $T$ is the length of the horizon for the repeated game. In each period $t$, let $\PiZero$ be the
payoff obtained by the primary user if the secondary user exits; then $\PiZero = \Pi_1(SP, X) = \log\left(1 + \frac{P}{2N_{0}}\right).$
As before, we do not assume a cost of power for the primary
user; this does not affect the results presented in this
section. On the other hand, the secondary user is considered to be
myopic: its objective is to maximize its one period payoff. If the
secondary user decides to exit the game, it obtains zero rate with no cost
of power, so $\Pi_2(SP, X) = 0$.

The per-period payoffs of the primary player and the secondary player are thus given as:
\begin{equation}
\setlength{\nulldelimiterspace}{0pt}
\Pi_{1}(\v{a}_{t}) = \left\{\begin{array}{ll}
\Pi_{0},\ & \text{if}\ \v{a}_{t} = (SP, X) \\
\Pi^{\text{share}}_{1},\ & \text{if}\ \v{a}_{t} = (SH, N) \\
\Pi^{\text{spread}}_{1},\ &\text{if}\ \v{a}_{t} = (SP, N);
\end{array}\right.
\setlength{\nulldelimiterspace}{0pt}
\Pi_{2}(\v{a}_{t}) = \left\{\begin{array}{ll}
0,\ & \text{if}\ \v{a}_{t} = (SP, X) \\
\Pi^{\text{share}}_{2} - kP,\ & \text{if}\ \v{a}_{t} = (SH, N) \\
\Pi^{\text{spread}}_{2} - kP,\ & \text{if}\ \v{a}_{t} = (SP, N).
\end{array}\right.
\label{Eq:PayoffMatrixPrimary}
\end{equation}
Here $\Pi^{\text{share}}_{i}$ and $\Pi^{\text{spread}}_{i}$ are
defined in 
\eqref{Eq:JointRates2}. Since
$\Pi_i^{\text{spread}}$ may be stochastic, $\Pi_i(\v{a}_t)$ may be
stochastic as well.

\subsection{Sequential Equilibrium of the Repeated Game}
\label{sec:ReputationResults}

In this section we study sequential equilibria of the repeated SBGI-E
game.  Note that all exogenous parameters are known by the primary
under Assumption \ref{As:informationSets}.  However, the secondary
does not know the channel gain $g_{21}$ of the primary, and instead
maintains a conditional distribution, or {\em belief}, of the value of
$g_{21}$ given the observed history $h_{2,t}$.  The secondary user
updates his beliefs in a Bayesian manner as the history evolves (see
footnote \ref{foot:seqeq}).

As shown in Proposition \ref{Thm:NE-SGI-E}, if $g_{12} < \tilde{g}_{12}$
for the single period SBGI-E game, the secondary user's entry decision
depends only on its realized channel gain $g_{12}$ and the cost of
power.  It follows that regardless of the secondary's beliefs, 
if $g_{12} < \tilde{g}_{12}$ the secondary user
either enters in every period or it stays out in every period; i.e.,
its strategy is independent of history.  Thus, in the
sequential equilibrium of the repeated game, each period follows the
sequential equilibrium of the single period game (cf. Proposition
\ref{Thm:NE-SGI-E}).  

More interesting behavior arises if $g_{12} > \tilde{g}_{12}$.
In this case, from Proposition \ref{Thm:NE-SGI-E}, the secondary user
prefers to enter if there is high probability the subchannels will be
shared by the primary, and prefers to exit otherwise.  Since the
primary's action depends on its gain $g_{21}$, in this case we
must calculate the secondary user's conditional distribution of
$g_{21}$ after each history $h_{2,t}$.  Belief updating can lead to  significant
analytical complexity, as the belief is infinite-dimensional (a distribution over a
continuous space).

However, the result of Proposition \ref{Thm:NE-SGI-E} suggests perhaps
some reduction may be possible: as noted there, the secondary's action
only depends on its belief about whether $g_{21}$ is larger or smaller
than $g^*$, which can be reduced to a scalar probability.  If we can
exhibit a sequential equilibrium of the repeated game in which the
primary's action only depends on whether $g_{21}$ is larger or smaller
than $g^*$ as well, then we can represent the secondary's belief by a
scalar sufficient statistic, namely the probability that $g_{21}$ is
larger than $g^*$.

Remarkably, we show that precisely such a reduction is possible, by
exhibiting a sequential equilibrium with the desired property.  Define
$\mu_{2,t}(h_{2,t}) = \P(g_{21} < g^{*}_{21} | h_{2,t})$ .  We exhibit
a sequential equilibrium where (1) the entry decision of the secondary
user in period $t$ is based only on this probability; and (2) the
primary's strategy is entirely determined by whether $g_{21}$ is
larger or smaller than $g^*$.  In this equilibrium $\mu_{t}$ will be a
sufficient statistic for the history of the play until time $t$.

The equilibrium we exhibit has the property that the primary user can
exploit the lack of knowledge of the secondary user. To illustrate
this point, consider a simplistic 2 period game. Assume that $g_{21} >
g^{*}$, so that in a single period game the primary prefers $SH$ to
$SP$ after entry by the secondary.  If the secondary enters in the
first period and the primary plays $SP$ (spread), the secondary may
mistakenly believe $g_{21}$ to be {\em small}---and thus expect its
payoff to be negative in the second period as well, and hence not
enter.  The primary thus obtains total payoff
$\Pi^{\text{spread}}_{1} + \PiZero$.  By contrast, if
the primary had shared in the first period, the secondary would
certainly have entered in the second period as well, and in this case
the primary obtains payoff $2\Pi^{\text{share}}_1$.  
It can be easily shown that if $g_{21} < 1$, then 
$\Pi^{\text{spread}}_{1} + \PiZero > 2\Pi^{\text{share}}_{1}$. Thus
under such conditions, the primary user can benefit by spreading even
though its single-period payoff is maximized by sharing.

The above example highlights the fact that the secondary user's lack
of information can be exploited by the primary user for its own
advantage. The primary user can masquerade and build a
\textit{reputation} as an ``aggressive'' player, thereby preventing entry by
the secondary user. Such ``reputation effects"' were first studied in
the economics literature \cite{KW_ECONTH1982, MR_ECONTH1982}, where
the authors show that the lack of complete information can lead to
such effects.\footnote{The lack of information can also be used to sustain
desirable equilibria, as shown in the case of the finitely repeated
prisoner's dilemma \cite{KMRW_ECONTH1982}. For a comprehensive
treatment of such reputation effects see
\cite{Mailath_samuelson_2006}.}

For the repeated SGI-E game with $g_{12} > \tilde{g}_{12}$, a
sequential equilibrium can be derived using an analysis closely following
\cite{KW_ECONTH1982}. Here (for notational simplicity) the periods are numbered in reverse
numerical order. Thus, $T$ denotes the first chronological period, and
$1$ the last.  We have the following theorem. 
\begin{theorem}
Suppose Assumptions
\ref{As:informationSets} and \ref{As:SignOfb} hold, and that $g_{12} >
\tilde{g}_{12}$.  Also assume that $\P(g_{21} < 1) = 1$. Let
$d$ be defined as in \eqref{Eq:ProbCutoff}. Then the following actions
and the belief update rule form a sequential equilibrium of the
finite horizon repeated SBGI-E game:
\begin{enumerate}
\item The secondary user in period $t$ enters the system if 
  $\mu_{2,t}(h_{2,t}) < d^t$, and it exits the system if $\mu_{2,t}(h_{2,t}) > d^{t}$.
  If $\mu_{2,t} = d^{t}$, the secondary user enters with probability
  $\lambda$, and exits with probability $ 1- \lambda$, where:
\[
\lambda = 2 - \frac{\PiZero - \Pi^{\text{spread}}_{1}}{\PiZero - \Pi^{\text{share}}_{1}}.
\]

\item If $g_{21} < g^*$, the primary user always spreads its power.
 If $g_{21} > g^*$, then after entry by a secondary user in period
  $t > 1$, the primary user always spreads if $\mu_{t} \geq d^{t-1}$, and
  otherwise randomizes, with the probability of spreading equal to:
\begin{align}
\gamma = \frac{\mu_{t}}{(1-\mu_{t})}\frac{(1-d^{t-1})}{d^{t-1}}.
\end{align}
For $t=1$, a primary user with $g_{21} > g^*$ always shares.
\item The beliefs $\mu_{t}$ are updated as follows:
\begin{equation}
\setlength{\nulldelimiterspace}{0pt}
\mu_{2,t}(h_{2,t})=\left\{\begin{array}{ll}
\mu_{2,t+1}(h_{2,t+1}),\ & \text{if}\ a_{t+1} = (X, \phi);\\
\max\{d^{t}, \mu_{2,t+1}(h_{2,t+1})\},\ & \text{if}\ a_{t+1} = (N, SP)\
\text{and}\ \mu_{2,t+1}(h_{2,t+1}) > 0;\\
0,\ & \text{if}\ a_{t+1} = (N, SH)\ \text{or}\ \mu_{2,t+1}(h_{2,t+1}) = 0.%
\end{array}\right.
\end{equation}
\end{enumerate}
\label{Thm:ReputationResult}
\end{theorem}

The proof of the above theorem follows steps similar to those in \cite{KW_ECONTH1982}. An outline of the proof for $T = 2$ is given in the appendix. Note that $d < 1$, so $d^t$ increases as the game progresses (since $t
= T, \ldots, 1$).  The first period in which the secondary will enter
is when its initial belief $\rho = \P(g_{21} < g^*)$ first falls below
$d^{t}$; thus even if $g_{21} > g^*$, entry is deterred from $T$ up to (approximately) $t^{*} =
\frac{\log(\rho)}{\log(d)}$.   It is important to note that this {\em
  never} happens in a complete information game: if the secondary knew
$g_{21} > g^*$ it would enter in {\em every} time period, and the
primary would always share.

We conclude by noting that it is straightforward to show that
equilibria may be inefficient: for fixed $\rho$ and $d$, $t^*$ is
constant, so as $T$ increases the number of periods in which entry is
deterred increases without bound.  For parameter values where it would
have been better to allow both users to transmit in each period, the
resulting equilibrium is clearly inefficient.

%% file: Conclusion.tex
\section{Conclusion}
\label{sec:conclusions}
We have studied distributed resource allocation in wireless systems via static and sequential Gaussian interference
games of incomplete information.  Our analysis shows that equilibria
of these Bayesian games exhibit significant differences from their
complete information counterparts.  In particular, we have shown in two settings that
static Gaussian interference games have a unique, potentially
inefficient equilibrium where all users spread their powers.  More
dramatically, in repeated sequential games, we have shown that the
lack of channel information can lead to reputation effects. Here the
primary user has an incentive to alter its power profile to keep
incoming secondary users from entering the system.

%% file: Appendix.tex
\appendix


\section{Sequential Interference Games with Two-Sided Uncertainty}
\label{sec:appendix}

The model for the two stage sequential Bayesian Gaussian interference
(SBGI) game described in Section \ref{sec:2StageGame} assumed that the
primary user is aware of the channel gain $g_{12}$. In this appendix
we analyze the case where both users are aware of only their own
incident channel gains. We refer to this case as \textit{two-sided
uncertainty}.  In Section \ref{sec:SBGItwosided}, we analyze a two
stage sequential game analogous to Section \ref{sec:2StageGame}.  In Section
\ref{sec:SBGIEtwosided}, we extend this analysis to include an entry
stage as well.

\subsection{A Two Stage Sequential Game}
\label{sec:SBGItwosided}

As before, we solve for the equilibrium path using backward
induction.  Whenever the primary user chooses $SP$, the best response of the
secondary user is to choose $SP$ regardless of the value of $g_{12}$;
the reasoning is identical to the proof of Lemma \ref{Lemma:SBGI}.
Therefore, if the
primary user chooses the action $SP$, its payoff is
$\PiPrimary^{\text{spread}}$ which is given by 
\begin{align}
\PiPrimary^{\text{spread}}(g_{21}) = \log\left(1 + \frac{P/2}{N_{0} +
g_{21}P/2}\right). 
\label{Eq:PrimaryPayoffSP}
\end{align}
(Note that we now explicitly emphasize the dependence of the payoff
on the channel gain $g_{21}$.)

If the primary user decides to share the bandwidth, i.e., chooses the
action $SH$, then depending upon the value of $g_{12}$ the secondary
user will choose between share ($SH$) or spread ($SP$). From Lemma
\ref{Lemma:SBGI} we know that if $g_{12} < 1/2$, the secondary user
would prefer to spread its power even though the primary user shares
the subchannels. Let $\kappa  = \P(g_{12} <
1/2)$. Since the primary user is unaware of $g_{12}$, its expected
payoff (denoted by $\overline{\Pi}_{1}^{\text{share}}$) if it shares
the bandwidth is given by
\begin{align}
\overline{\Pi}_{1}^{\text{share}}(g_{21}) = \frac{1 - \kappa}{2} \log\left( 1 + \frac{P}{N_{0}}\right) + \frac{\kappa}{2}\log\left(1 + \frac{P}{N_{0} + g_{21}P/2}\right)
\label{Eq:PrimaryExpPayoffSH}
\end{align}

We have the following proposition.

\begin{proposition}
In the first stage of the SBGI game with two-sided uncertainty, there
exists a threshold $\hat{g}_{21}(\kappa) > 0$ (possibly infinite) such
that if $g_{21}(\kappa) < \hat{g}_{21}$, the primary user always
spreads its power, and if $g_{21} > \hat{g}_{21}(\kappa)$, the primary
user always shares the subchannels.
\label{Prop:SBGI-2Sided}
\end{proposition}

\begin{proof}
To decide whether to spread or share, the primary user needs to
compare $\PiPrimary^{\text{spread}}(g_{21})$ to
$\overline{\Pi}_{1}^{\text{share}}(g_{21})$.  We begin by establishing
that $\Delta(g_{21}) = \PiPrimary^{\text{spread}}(g_{21}) -
\overline{\Pi}_{1}^{\text{share}}(g_{21})$ is strictly decreasing in
$g_{21}$.  To see this, note that if we define $y(g_{21}) = P/(N_0 + g_{21}
P/2)$, then
\[ \Delta'(g_{21}) =  y'(g_{21}) \left( \frac{1/2}{1 + y(g_{21})/2} - \frac{\kappa/2}{1 +
y(g_{21})}\right). \] 
Since $y'(g_{21}) < 0$, and $0 \leq \kappa \leq 1$, we conclude that
$\Delta'(g_{21}) < 0$ for all $g_{21}$; i.e., $\Delta(g_{21})$ is
strictly decreasing in $g_{21}$, as required.

When $g_{21} = 0$, we have:
\begin{align*}
\overline{\Pi}_{1}^{\text{share}}(0) & = \frac{1}{2}\log\left(1 +
\frac{P}{N_0}\right);\\ 
\PiPrimary^{\text{spread}}(0) &= \log\left( 1+ \frac{P}{2N_0}\right).
\end{align*}
Since $\log(1+x)$ is a strictly concave function of $x$,
$\PiPrimary^{\text{spread}}(0) >
\overline{\Pi}_{1}^{\text{share}}(0)$. 
However, if $g_{21}$ is large, the payoffs to the primary
user in the two cases are given by:
\begin{align*}
\lim_{g_{21} \rightarrow \infty} \overline{\Pi}_{1}^{\text{share}} & = \frac{1-\kappa}{2}\log\left(1 + \frac{P}{N_0}\right);\\
\lim_{g_{21} \rightarrow \infty} \PiPrimary^{\text{spread}} &= 0.
\end{align*}
Thus, $\lim_{g_{21} \to \infty} \overline{\Pi}_{1}^{\text{share}} \geq
\lim_{g_{21} \to \infty} \PiPrimary^{\text{spread}}$.  If $\kappa =
1$, then $\Delta(g_{21}) > 0$ for all $g_{21}$, and thus user 1 always
spreads; i.e., $\hat{g}_{21}(\kappa) = \infty$.  Otherwise, there
exists a unique finite threshold $\hat{g}_{21}(\kappa) > 0$ determined
by the equation $\Delta(g_{21}) = 0$, as required.
\end{proof}


For the secondary user, the best response is to spread the power if
the primary user spreads its power. However, if the primary user
decides to share the subchannels, the secondary user will share if
$g_{12} > 1/2$; otherwise it spreads its power. This completely determines
the sequential equilibrium for the SBGI game with two-sided
uncertainty.

\subsection{A Sequential Game with Entry}
\label{sec:SBGIEtwosided}

We now consider the sequential Bayesian Gaussian interference with
entry (SBGI-E) game when both users are only aware of their own
incident channel gains; i.e., player $i$ only knows $g_{-i,i}$.  As
before, the secondary user first decides whether to enter or not; if
the secondary user exits, then the primary uses the entire band
without competition.  If the secondary enters, then play proceeds
as in the SBGI game of the preceding section.  Further, if the secondary
chooses to enter, it incurs a cost of power denoted by $kP$.

To describe a sequential equilibrium for the SBGI-E game with
two-sided uncertainty, we define $\PiSecondary^{\text{(share, a)}}$ to
be the rate of the secondary user when the primary user shares the
subchannels and the secondary user chooses the action $a \in \{SH,
SP\}$. Thus, we have:
\begin{align}
\PiSecondary^{\text{(share, share)}} &= \PiSecondary^{\text{share}} = \frac{1}{2}\log\left(1 + \frac{P}{N_0}\right),\\
\PiSecondary^{\text{(share, spread)}} &= \frac{1}{2}\log\left(1 + \frac{P}{2N_0}\right) + \frac{1}{2}\log\left( 1+ \frac{P/2}{N_0 + g_{12}P}\right).
\label{Eq:PayoffSecTwoSided}
\end{align}
We also recall the rate to the secondary user when both users spread
their powers, denoted $\PiSecondary^{\text{spread}}$
(cf. \eqref{Eq:JointRates2}):
\[ \PiSecondary^{\text{spread}} = \log \left( 1 + \frac{P/2}{N_0 +
g_{12} P/2} \right). \]

In this game, we must be particularly careful about how uncertainty
affects sequential decisions.  In particular, since the entry decision
of the secondary user will depend on the gain $g_{12}$, the primary
user {\em learns} about the value of $g_{12}$ from the initial action
of the secondary; this is modeled through the updated {\em belief} of
the primary user, i.e., its conditional distribution of $g_{12}$ given
the initial action of the secondary.  Post-entry, the play proceeds as in
the SBGI game with two-sided uncertainty considered in the previous
section.

The following proposition formally describes sequential equilibria for the
SBGI-E game with two-sided uncertainty; for simplicity, we assume the
gain distributions have full support on $(0,\infty)$, but this is unnecessary.
\begin{proposition}
\label{Thm:SeqEqTwoSided}
Assume that the channel gains $g_{12}$ and $g_{21}$ have strictly
positive densities on $(0,\infty)$.  For the sequential Bayesian
Gaussian interference game with entry (SBGI-E), any sequential
equilibrium consists of {\em threshold strategies} for both the
primary and secondary user; i.e., there exists a threshold
$\hat{g}_{12}$ such that the secondary user enters if $g_{12} <
\hat{g}_{12}$, and exits if $g_{12} > \hat{g}_{12}$; and post-entry,
there exists a $\hat{\kappa}$ such that the primary user spreads if
$g_{21} < \hat{g}_{21}(\hat{\kappa})$, and shares if $g_{21} >
\hat{g}_{21}(\hat{\kappa})$ (where $\hat{g}_{21}(\cdot)$ is the
threshold of Proposition \ref{Prop:SBGI-2Sided}).  In the third stage, if the primary
user spreads its power, the secondary user also spreads its power
regardless of the value of $g_{12}$. However, if the primary user
shares the subchannels, the secondary user spreads its power if
$g_{12} < 1/2$ and it shares otherwise.

The belief $\hat{\kappa}$ is computed using Bayes' rule:
\begin{align}
\label{Eq:BeliefUpdatePrimary}
\hat{\kappa} = \frac{\kappa\P(N\ |\ g_{12} < 1/2)}{\kappa\P(N\ |\ g_{12} < 1/2) + (1-\kappa)\P(N\ |\ g_{12} > 1/2)},
\end{align}
where $N$ denotes the ``entry'' action of the secondary user, and
$\kappa = P(g_{12} < 1/2)$ (the initial belief of the primary).
\end{proposition}

\begin{proof}
We first show that in equilibrium, the primary must have a threshold
strategy of the form specified.  Note that if the secondary player
exits in the first stage, the primary user has no action to take and
the game ends. If the secondary player enters, the primary user
updates its belief about $g_{12}$ via Bayes' rule as given in
\eqref{Eq:BeliefUpdatePrimary}, given the entry strategy of the
secondary. Here the probabilities are with respect to the uncertainty
in $g_{12}$.

Now suppose that in equilibrium, the post-entry belief of the primary
user is fixed as $\hat{\kappa}$.  We consider the entry decision of the
secondary user. The secondary user's decision between entry or exit
depends on the post-entry action taken by the primary user. From
Proposition \ref{Prop:SBGI-2Sided}, we know that post-entry, the
primary user will spread (resp., share) if $g_{21}$ is less than
(resp., greater than) the threshold $\hat{g_{21}}(\hat{\kappa})$.  (Here
the threshold $\hat{g_{21}}$ depends on the post-entry belief
$\hat{\kappa}$ of the primary user.)  Thus, the decision taken by the
secondary user in the first stage depends on its initial belief
$\alpha = \P(g_{21} < \hat{g_{21}}(\hat{\kappa}))$.  Since
$\hat{g_{21}}(\hat{\kappa}) > 0$ from Proposition \ref{Prop:SBGI-2Sided}, it follows
that $\alpha > 0$.

After entry, with probability $\alpha$, the primary user
spreads its power; since the best response of the secondary user to
the spreading action by the primary user is to also spread, the payoff
of the secondary user in this case is $\PiSecondary^{\text{spread}}$. With
probability $1 - \alpha$, the primary user shares
the subchannels. Conditioned on the sharing action by the primary
user, the secondary user will spread its power if $g_{12} < 1/2$ and
will share otherwise. Thus, the expected payoff of the secondary user
upon entry is given by the function $h(g_{12}, \alpha)$, defined
as follows:
\begin{equation}
\setlength{\nulldelimiterspace}{0pt}
h(g_{12}, \alpha) = \left\{ \begin{array}{ll}
\alpha\PiSecondary^{\text{spread}} +
(1-\alpha)\PiSecondary^{\text{(share, spread)}},\ & \text{if}\ g_{12} \leq 1/2,
\\ 
\alpha\PiSecondary^{\text{spread}} +
(1-\alpha)\PiSecondary^{\text{(share, share)}},\ \text{if}\ g_{12} > 1/2. 
\end{array}\right.
\label{Eq:SecExpectedPayoff}
\end{equation}
It is easy to check that $h(g_{12}, \alpha)$ is continuous; further,
for fixed $\alpha > 0$, $h(g_{12}, \alpha)$ is a strictly decreasing
function of $g_{12}$.  The secondary user will enter if its expected
payoff $h(g_{12}, \hat{\alpha})$ is greater than its cost of power $kP$. 
If, $h(0, \alpha) < kP$, let $\hat{g}_{12} = 0$.   Similarly, if $\lim_{g_{12}
\to \infty} h(g_{12}, \alpha) > kP$, we define $\hat{g}_{12}(\alpha) =
\infty$.   Otherwise there exists a unique value of
$\hat{g}_{12} \in (0, \infty)$ with $h(\hat{g}_{12}, \alpha) = kP$.  Thus,
the secondary user will enter if $g_{12} < \hat{g}_{12}$, and exit if
$g_{12} > \hat{g}_{12}$.  This concludes the proof.
\end{proof}

\section{Sequential Equilibrium for a Two Period Repeated Game}
In this appendix, for completeness we give an outline of the proof of Theorem
\ref{Thm:ReputationResult} for the two period repeated SBGI-E game
with \textit{single-sided uncertainty}. The arguments given below are
based on those given in \cite{KW_ECONTH1982}, and we refer the reader
to that paper for details and extensions.

To analyze the two period repeated SBGI-E game, we use backward
induction. We number the periods in reverse numerical order. Thus,
period $1$ is the last period of the game, and period $T$ is the first
period of the game; period $t$ follows period $t+1$. To specify the
equilibrium, we need to specify the actions of the secondary user and
the primary user in all periods and after all possible values of
histories and beliefs. Note that if $g_{21} < g^{*}$, the primary user
always spreads its power. So we need to specify the action of the
primary user for the case when $g_{21} > g^{*}$; we refer to this type
of primary user as a {\em high-gain primary user}. For the remainder
of this discussion, we only specify the actions of the high-gain
primary user. (We ignore the $g_{21} = g^{*}$ case since channel gains
have continuous densities). Also note that although the beliefs of the
secondary user are history dependent, we suppress the history
dependence of the beliefs for notational simplicity.

\begin{itemize}

\item \textbf{Period 1}: In period $1$ (which is the last period), if
the secondary user does not enter, the high-gain primary user has no action to
take. However, if the secondary user decides to enter, the high-gain primary
user will share the subchannels, since sharing the
subchannels is the best response of the high-gain primary user to an entry by
the secondary user. To decide between entry and exit, the secondary
user takes into account its belief $\mu_{2,1}$ about the channel gain
$g_{21}$ of the primary user. The secondary user will enter if its
expected payoff in period $1$ is greater than its payoff if it
exits. This happens if 
\begin{align}
\mu_{2,1}(\PiSecondary^{\text{spread}} - kP) + (1 - \mu_{2,1})(\PiSecondary^{\text{share}} - kP) > 0 \quad \implies \ \mu_{2,1} < \frac{\PiSecondary^{\text{share}} -kP}{\PiSecondary^{\text{share}} - \PiSecondary^{\text{spread}}} = d.
\end{align}
Here $d$ is defined as in \eqref{Eq:ProbCutoff}. Thus, in equilibrium
the secondary user enters (N) if its current belief $\mu_{2,1} < d$,
it exits the system if $\mu_{2,1} > d$, and it is indifferent if
$\mu_{2,1} = d$.

To find the current belief $\mu_{2,1}$, the secondary user observes
the history of the play. If in period~$2$, the secondary user exits
the game $(X)$, no new information about the primary user is
learned. Thus if $h = (X)$, we have $\mu_{2,1} = \mu_{2,2}$. Since at
period $2$, no history has been observed, we have $\mu_{2,2} = \rho$,
which is the initial belief of the secondary user. However, if the
secondary user in period $2$ enters, the belief of the secondary user
in period $1$ would depend upon the action taken by the primary user
in period $2$. If the primary user shares the subchannels in period
$1$, then it is certain that $g_{21} > g^{*}$ and hence $\mu_{2,1} =
0$.

When the history $(N, SP)$ is observed, the secondary user uses Bayes'
rule to update its belief. Let $\gamma$ denote the probability that
the primary user would spread its power even if $g_{21} > g^{*}$. Then
the total probability that the primary user would spread (in period 2)
is $\mu_{2,2} + (1 - \mu_{2,2})\gamma$. Bayes' rule then
implies that the belief in period $1$ is given as
\begin{align}
\mu_{2,1} = \frac{\mu_{2,2}}{\mu_{2,2} + (1 - \mu_{2,2})\gamma}.
\label{Eq:BeliefUpdate}
\end{align}
Here the numerator is the probability that $g_{21} < g^{*}$ in period $2$.

\item \textbf{Period 2}: For the high-gain primary user in period $2$, the
action it takes in this period determines the history for period $1$,
and hence the action taken by the secondary user. The high-gain primary user
thus needs to conjecture the behavior of the secondary user in period
$1$ to decide its action. Note that the belief of the secondary user
in period $2$ is same as the initial belief, i.e., $\mu_{2,2} = \rho$.

If the secondary user does not enter the game at this period, the
primary user has no action to take. However, if the secondary user
enters, the high-gain primary user has to choose between the actions $SH$ or
$SP$. It chooses this action so as to maximize its expected total
payoff in the two periods. Here the expectation is over the randomness
in the action taken by the secondary user in period $1$\footnote{We
are assuming that secondary user $2$ has already entered, so the only
unknown factor is the action of the secondary user in period 1.}.

First note that in equilibrium $\gamma > 0$.  To see this, let us
assume otherwise, i.e., $\gamma = 0$. This implies that in
equilibrium, if the secondary user enters, the high-gain primary user does not
spread. Then the high-gain primary user's total payoff in $2$ periods is
$2\PiPrimary^{\text{share}}$. However, if the high-gain primary user spreads in
period $2$, then the secondary user in period $1$ has $\mu_{2,1} = 1$
(see \eqref{Eq:BeliefUpdate}) and hence it does not enter. In this
case, the total payoff to the high-gain primary user is
$\PiPrimary^{\text{spread}} + \PiZero$ which is greater than
$2\PiPrimary^{\text{share}}$ since $g_{21} < 1$. Hence the high-gain primary
user has a profitable deviation in equilibrium with $\gamma =
0$. Thus, in equilibrium $\gamma > 0$.

We consider two different cases. First suppose that $\mu_{2,2}
= \rho \geq d$. In this case, regardless of the strategy of the
primary user in period $2$, we have
\begin{align*}
\mu_{2,1} > \rho \geq d.
\end{align*}
Hence, the secondary user in period $1$ would not enter after seeing
the history of $SP$. So if the high-gain primary user in period $2$ takes the
action $SP$, the total payoff is $\PiPrimary^{\text{spread}} +
\PiZero$. On the other hand taking the action $SH$ would cause the
secondary user in period $1$ to enter, and hence the total payoff
would be $2\PiPrimary^{\text{share}}$. Since
$\PiPrimary^{\text{spread}} + \PiZero >2\PiPrimary^{\text{share}}$,
the best response for the high-gain primary user in period $2$ (if the secondary
enters and $\rho > d$) is to spread the power.

The second case is when $\mu_{2,2} = \rho < d$. In this case, we first
note that $\gamma < 1$. If we assume that $\gamma = 1$, then
$\mu_{2,1} = \mu_{2,2} < d$ and hence the secondary user in period $1$
would always enter and the high-gain primary would spread (since
$\gamma = 1$). But we know that in period $1$, the best response of
the high-gain primary user to an entry is to share. Hence $\gamma <
1$. This implies that if $\mu_{2,2} < d$, then $0 < \gamma < 1$. Thus
the high-gain primary user randomizes its policy over $SP$ and $SH$. This is
only possible if the secondary user in period $1$ also randomizes over
entry and exit. Let us denote the probability of secondary user
entering in period $1$ under these conditions as $\lambda$. Since the
high-gain primary user in period $2$ is indifferent between spreading and
sharing its expected payoff in both cases is the same. This gives
\begin{align*}
\PiPrimary^{\text{spread}} &+ \lambda\left(\PiPrimary^{\text{share}}\right) + (1-\lambda)\PiZero = 2\PiPrimary^{\text{share}} \\
\implies &\lambda = 2 - \frac{\PiZero -\PiPrimary^{\text{spread}}}{\PiZero - \PiPrimary^{\text{share}}}.
\end{align*}
Also, since the secondary user in period $1$ is indifferent between
entry and exit, its belief is $\mu_{2,1} = d$. This gives  
\begin{align*}
\frac{\mu_{2,2}}{\mu_{2,2} + (1-\mu_{2,2})\gamma} = d \quad \implies \quad \gamma = \frac{\mu_{2,2}}{(1-\mu_{2,2})}\frac{1-d}{d}.
\end{align*}

To determine the action of secondary user at period $2$, we note that
if it exits its payoff is $0$. Now if $\mu_{2,2} = \rho \geq d$, the
high-gain primary user spreads with probability $1$, hence the best response for
the secondary user is to exit. However, if $\mu_{2,2} = \rho < d$,
then the secondary user's expected payoff is
\begin{align*}
(\PiSecondary^{\text{spread}})(\rho + (1-\rho)\gamma) + (1-\rho)(1-\gamma)\PiSecondary^{\text{share}}.
\end{align*}
If the above expected payoff is less than $0$, the secondary user does
not enter. Using the value of $\gamma$ from above we get that if $\rho
> d^{2}$, the secondary user exits. However, if $\rho < d^{2}$, the
secondary user enters, and at equality the secondary user is
indifferent. This completely specifies the sequential equilibrium for
the two period repeated SBGI-E game.
\end{itemize}

The extension to an arbitrary finite horizon repeated game is similar
to the arguments given above and we refer the reader to
\cite{KW_ECONTH1982} for the detailed proof.